\documentclass[journal,12pt,onecolumn]{IEEEtran}

\usepackage{amsmath,amsfonts}
\usepackage{amssymb,verbatim,amsfonts,amsmath,amsthm}
\usepackage{geometry}
\usepackage{cite}
\usepackage{hyperref}
\usepackage{bm}
\usepackage[all]{xy}

\usepackage{longtable}
\usepackage{booktabs}
\usepackage{graphicx}
\usepackage[figuresright]{rotating}

\newtheorem{theorem}{Theorem}
\newtheorem{lemma}{Lemma}
\newtheorem{corollary}[theorem]{Corollary}

\newtheorem{remark}{Remark}
\newtheorem{definition}{Definition}

\begin{document}

\title{On Permutation Groups of Cyclic Codes over Finite Fields}
\author{Junjie Huang, Jicheng Ma and Chang-An Zhao
\thanks{J. Huang is with the Department of Mathematics, School of Mathematics, Sun Yat-sen University, Guangzhou 510275, P.R.China. (e-mail: huangjj76@mail2.sysu.edu.cn).}
\thanks{J. Ma is with the School of Mathematics and Artificial Intelligence, Chongqing University of Arts and Sciences, Chongqing 402160, P.R.China. (e-mail: ma\_jicheng@hotmail.com).}
\thanks{C.-A. Zhao is with the School of Mathematics, Sun Yat-sen University, Guangzhou 510275, P.R.China, and also with the Guangdong Key Laboratory of Information Security, Guangzhou 510006, P.R.China (e-mail: zhaochan3@mail.sysu.edu.cn).}

}

\maketitle

\begin{abstract}
   The permutation groups of cyclic codes are widely applicable in determining the weight distribution of codes, decoding theory and various other areas. In this paper, by employing two distinct matrix representations, we can relate cyclic codes with very long lengths and special generator polynomials to those with prime lengths. Consequently, we mainly determine the permutation groups of certain cyclic codes over $\mathbb{F}_{r^\alpha}$ with lengths $hp$, $r^mp^n$ and $pq$ and special generator polynomials where $h$ is a positive integer and $p$, $q$ and $r$ are distinct prime numbers. For length $pq$, we manage to provide the permutation groups of cyclic codes with generator polynomials $Q_{pq}(x)$(the $pq$-th cyclotomic polynomial) or others, which seems to be the first work about permutation groups of cyclic codes with generator polynomials that are factors of $x^{pq}-1$ but not factors of $x^p-1(\text{or }x^q-1)$.
\end{abstract}


\section{Introduction}

In coding theory, cyclic codes have attracted considerable attention due to their efficient encoding and decoding algorithms~\cite{Blahut_2003}. Beyond this, the permutation groups of cyclic codes also hold significant theoretical and practical importance. The permutation groups of given codes can be applied to determine the weight distribution of codes~\cite{macwilliams1977}, decoding theory~\cite{huffman1998,9492151,1013146}, among other areas. Therefore, investigating the permutation groups of cyclic codes serves as a crucial bridge connecting abstract algebraic coding theory with practical engineering applications. Over the past decades, the permutation groups of linear codes have been extensively studied~\cite{PHELPS198645,BERGER19931,GHORPADE201380,9685181}. But so far, the permutation groups of only a few classes of linear codes have been classified, such as Hamming codes, Reed-Muller codes, etc~\cite{huffman1998}.

\subsection{Known Results}

In 1987, D{\"u}r determined the automorphism groups and permutation groups of Reed-Solomon codes and extended Reed-Solomon codes~\cite{DUR198769}. In 1997, Zanotti characterized the automorphism groups of a particular class of irreducible cyclic codes~\cite{ZANOTTI1997303}. In 1999, Berger and Charpin classified the automorphism groups of BCH codes affine-invariant codes~\cite{bergerAutomorphismGroupsBCH1999}. In 2007, Joyner et al. classified those finite groups that can arise as the automorphism group of an AG code, especially generalized Reed-Solomon codes, and gave an explicit description of how these groups appear~\cite{doi:10.1142/9789812772022_0008}. In 2010, Bienert and Klopsch studied the permutation groups of binary cyclic codes and classified all binary cyclic codes with primitive permutation groups~\cite{bienertAutomorphismGroupsCyclic2010}. In 2013, Guenda and Gulliver classified the permutation groups of cyclic codes with prime length  over a finite field and described the Sylow $p$-subgroup of the permutation groups of cyclic codes of length $p^m$. In 2021, Geiselhart et al. proposed a novel and effective permutation group decoder, called the automorphism ensemble (AE) decoder, based on the automorphism group of the given code~\cite{9492151}. Subsequently, Geiselhart et al. studied the automorphism group of polar codes and proved that the block lower-triangular affine group is contained therein~\cite{9518184}. Also in this year, Li et al. proved that the block lower-triangular affine group is the full affine automorphism group of polar codes~\cite{9685181}. In 2022, Liu et al. classified ${\rm PGL}(2,p^m)$-invariant codes over $\mathbb{F}_{p^h}$ with length $p^m+1$ and obtained infinite families of narrow-sense BCH codes admitting a $3$-transitive automorphism group~\cite{9666868}. In 2023, Hollmann investigated the permutation groups of general (repeated-root) cyclic codes and irreducible cyclic codes~\cite{HOLLMANN2023102146}. In 2025, Ma and Yan investigated the permutation groups of cyclic codes over $\mathbb{F}_2$ with lengths $2p$ and $p^n$ and irreducible generator polynomials by using a matrix presentation technique of cyclic codewords~\cite{maAutomorphismGroupsBinary2025}.

\subsection{Main Results}
In this paper, inspired by~\cite{maAutomorphismGroupsBinary2025}, we investigate the permutation groups of cyclic codes over arbitrary finite fields $\mathbb{F}_{r^\alpha}$ with specific lengths. Let $g(x)\in\mathbb{F}_{r^\alpha}[x]$ be a factor of $x^p-1$. Our main results of this paper are summarized as below.
\begin{itemize}
    \item For a positive integer $h$ and a prime number $p$, we completely determine the permutation group of $\mathcal{C}_{hp,g(x)}$(see Theorem~\ref{L_hp}). 
    \item Let $n$ be a positive number and $m$ be a non-negative number. For any $0\le u\le m$ and $0\le v\le n-1$, we completely determine the permutation group of $\mathcal{C}_{r^mp^n,g(x^{r^up^v})}$ and $\mathcal{C}_{hr^mp^n,g(x^{r^up^v})}$(see Theorem~\ref{L_rp} and Corollary~\ref{L_rp_col}). 
    \item For distinct prime numbers $p\neq r$ and $q\neq r$, we completely determine the permutation groups of $\mathcal{C}_{pq,Q_{pq}(x)}$, $\mathcal{C}_{pq,(x-1)Q_p(x)Q_q(x)}$ and $\mathcal{C}_{hpq,Q_p(x)Q_q(x)}$(see Theorem~\ref{L_pq} and Corollary~\ref{L_pq_col}). Moreover, it is noteworthy that, to our knowledge, it seems the first time that the permutation groups of cyclic codes whose generator polynomials are factors of $x^{pq}-1$ but not factors of $x^p-1(\text{or }x^q-1)$ has been obtained.
\end{itemize}

\subsection{Organization}
The paper is organized as follows. In Section~\ref{pre}, we briefly review cyclic codes, the permutation groups of cyclic codes, the notion of wreath product and provide the definition of matrix representations of cyclic codes. In Section~\ref{Per}, we determine the permutation groups of certain cyclic codes with lengths $hp$, $r^mp^n$ and $pq$ and special generator polynomials by employing two distinct matrix representations. In Section~\ref{Con}, we give a summary of the entire article.

\section{Preliminaries}\label{pre}

In this section, we provide a concise review of some preliminaries related to cyclic codes, the permutation groups of cyclic codes, matrix representations of cyclic codes and the notion of wreath product. 

\subsection{Cyclic Codes}

Let $\mathbb{F}_{r^\alpha}$ be a finite field where $r$ is a prime number and $\alpha$ is a positive integer. A $r^\alpha$-ary $[n,k,d]_{r^\alpha}$ linear code $\mathcal{C}$ is a linear subspace of $\mathbb{F}_{r^\alpha}^n$ with dimension $k$ and minimum distance $d$. A vector $c = (c_0,\cdots,c_{n-1})\in \mathcal{C}$ is called a codeword of $\mathcal{C}$. The formal definition of cyclic codes is given as follows.

\begin{definition}\cite[Def. 7.1.1]{Ling_Xing_2004}
    A subset $S$ of $\mathbb{F}_{r^\alpha}^n$ is cyclic if $(c_{n-1},c_0,c_1,\cdots,c_{n-2})\in S$ whenever $(c_0,c_1,\cdots,c_{n-1})\in S$. A linear code $\mathcal{C}$ is called a cyclic code if $\mathcal{C}$ is a cyclic set. 
\end{definition}

Consider the following correspondence:
\begin{align}\label{cyc_poly_cor}
    \pi : \mathbb{F}_{r^\alpha}^n&\rightarrow \mathbb{F}_{r^\alpha}[x]/(x^n-1),\notag\\
    (c_0,c_1,\cdots,c_{n-1})&\mapsto c_0 + c_1x + \cdots + c_{n-1}x^{n-1}.
\end{align}
It is well known that a nonempty subset $\mathcal{C}$ of $\mathbb{F}_{r^\alpha}^n$ is a cyclic code if and only if $\pi(\mathcal{C})$ is a principle ideal of $\mathbb{F}_{r^\alpha}[x]/(x^n-1)$. Moreover, the unique monic polynomial $g(x)$ of the least degree of the ideal $\pi(\mathcal{C})$ must be a divisor of $x^n-1$ and is called the generator polynomial of $\mathcal{C}$, and $h(x) = \frac{x^n-1}{g(x)}$ is referred to as the check polynomial of $\mathcal{C}$. In this paper, we denote by $\mathcal{C}_{n,g(x)}$ the cyclic code of length \(n\) with generator polynomial $g(x)\in\mathbb{F}_{r^\alpha}[x]$. 

The Euclidean inner product on $\mathbb{F}_{r^\alpha}$ is defined by
\[
    \langle c,d \rangle = \sum_{i=0}^{n-1} c_id_i,
\]
where $c = (c_0,\cdots,c_{n-1})$ and $d = (d_0,\cdots,d_{n-1})$. The (Euclidean) dual of a linear code $\mathcal{C}$ is 
\[
    \mathcal{C}^\perp = \{d\in\mathbb{F}_{r^\alpha}^n\mid \langle c,d \rangle = 0\text{ for all } c\in \mathcal{C}\}. 
\]
The dual code $\mathcal{C}^\perp$ of a cyclic code with generator polynomial $g(x)$ is also a cyclic code with the generator polynomial $h^\perp(x) = \frac{x^kh(x^{-1})}{h(0)}$, which is the reciprocal polynomial of the check polynomial $h(x)$ of $\mathcal{C}$.

\subsection{The Permutation Groups of Cyclic Codes}

Let $\mathcal{C}\subseteq \mathbb{F}_{r^\alpha}^n$ be a cyclic code with generator polynomial $g(x)\in\mathbb{F}_{r^\alpha}[x]$ and let $S_n$ be the symmetric group acting on $\{0,\cdots,n-1\}$. Then $S_n$ acts naturally on a codeword $c = (c_0,\cdots,c_{n-1})$ of $\mathcal{C}$, i.e.
\[
    c^\sigma = (c_{\sigma(0)},\cdots,c_{\sigma(n-1)}),
\]
where $\sigma$ is a permutation in $S_n$. The permutation group of a cyclic code $\mathcal{C}$ is defined by
\[
    \operatorname{Per}(\mathcal{C}) = \{ \sigma\in S_n\mid \mathcal{C}^\sigma = \mathcal{C} \}, 
\]
where $\mathcal{C}^\sigma = \{c^\sigma\mid c\in\mathcal{C}\}$. 

In~\cite{bienertAutomorphismGroupsCyclic2010}, Bienert and Klopsch gave a characterization of binary cyclic codes with primitive automorphism groups. In~\cite{guendaPermutationGroupsCyclic2013}, Guenda and Gulliver generalized the results of \cite{bienertAutomorphismGroupsCyclic2010} and established the following characterization on permutation groups of cyclic codes with prime length $p$.

\begin{lemma}\label{per_cyc_p}\cite[Thm. 3]{guendaPermutationGroupsCyclic2013}
    Let $\mathcal{C}$ be a cyclic code of length $p$ prime over $\mathbb{F}_{r^\alpha}$. Then $\operatorname{Per}(\mathcal{C})$ is a primitive group, and one of the following holds.
    \begin{enumerate}
        \item $\operatorname{Per}(\mathcal{C})$ is a solvable group of order $pm$ with $m$ a divisor of $p-1$ and $C_p\le \operatorname{Per}(\mathcal{C}) \le \operatorname{AGL}(1,p)$, with $p\ge 5$. Furthermore $\operatorname{Per}(\mathcal{C})$ contains a normal Sylow $p$-subgroup.
        \item $\operatorname{Per}(\mathcal{C}) = S_p$.
        \item If $p = r$ and $\alpha = 1$, then $\operatorname{Per}(\mathcal{C}) = \operatorname{AGL}(1,p)$.
        \item $\operatorname{Per}(\mathcal{C}) = \operatorname{PSL}(2,11)$ and $p = 11$ and $r = 3$. 
        \item $\operatorname{Per}(\mathcal{C}) = M_{23}$ and $p = 23$ and $r = 2$. 
        \item $\operatorname{Per}(\mathcal{C}) = \operatorname{P\Gamma L}(d,r^{d^b})$ where $b\in\mathbb{N}$, $d\ge 3$ is a prime number such that $\gcd(d,r^{d^b}-1)$ and $p = (r^{d^{b+1}} - 1)/(r^{d^b} - 1)$.
    \end{enumerate}
\end{lemma}

Note that, in Lemma~\ref{per_cyc_p}, the \( S_p \) denotes the symmetric group of degree \( p \), the \( C_p \) denotes the cyclic group of order \( p \), the \(\operatorname{AGL}(1, p)\) denotes the 1-dimensional affine general linear group of degree \( p \), the \(\operatorname{PSL}(2,11)\) denotes the 2-dimensional projective special linear group of degree \( 11 \), the \( \operatorname{P\Gamma L}(d,r^{d^b}) \) denotes the projective semi-linear group, and the \( M_{23} \) denotes the Mathieu group of degree 23.

\subsection{Matrix Representations of Cyclic Codes}

Let $\ell$ be a divisor of $n$. We represent the codeword $c = (c_0,\cdots,c_{n-1})$ of $\mathcal{C}$ by two types of matrix representations as follows.

\begin{definition}
    ~
    \begin{enumerate}
        \item For any $0\le i\le \frac{n}{\ell}-1$, we place the symbols $(c_{i\ell},\cdots,c_{i\ell+\ell-1})$ in the $i$-th row of the matrix. Therefore, the codeword $c$ of $\mathcal{C}$ can be represented by the following $\frac{n}{\ell}\times\ell$ matrix
        \[
            Mr_{\ell}(c) = \begin{pmatrix}
                c_0 & c_1 & \cdots & c_{\ell-1} \\
                c_{\ell} & c_{\ell+1} & \cdots & c_{2\ell-1} \\
                \vdots & \vdots & \ddots & \vdots \\
                c_{n-\ell} & c_{n-\ell+1} & \cdots & c_{n-1}
            \end{pmatrix}. 
        \]
        We define $Mr_{\ell}(c)$ as the first type of matrix representation for $\mathcal{C}$.
        \item For any $0\le i\le \frac{n}{\ell}-1$, we place the symbols $(c_{i\ell},\cdots,c_{i\ell+\ell-1})^T$ in the $i$-th column of the matrix. Therefore, the codeword $c$ of $\mathcal{C}$ can be represented by the following $\ell\times\frac{n}{\ell}$ matrix 
        \[
            Mc_{\ell}(c) = \begin{pmatrix}
                c_0 & c_{\ell} & \cdots & c_{n-\ell} \\
                c_{1} & c_{\ell+1} & \cdots & c_{n\ell-\ell+1} \\
                \vdots & \vdots & \ddots & \vdots \\
                c_{\ell-1} & c_{2\ell-1} & \cdots & c_{n-1}
            \end{pmatrix}. 
        \]
        We define $Mc_{\ell}(c)$ as the second type of matrix representation for $\mathcal{C}$.
    \end{enumerate}
\end{definition}
\begin{remark}
    From the perspective of linear algebra, the matrix representations of cyclic codes above are merely simple algebraic representations, and the two representations are mutual matrix transposes of each other. However, these two matrix representations prove highly useful for determining the permutation group of the cyclic codes, and they will be used in two distinct cases. 
\end{remark}

\subsection{Wreath Product}

The notion of a wreath product arises very naturally in the study of imprimitive groups. For more details, the reader can refer to~\cite{dixon_permutation_1996}.

\begin{definition}
    Let \( A \) be a group and let \( H \) be a group acting on a set \( \Omega \). The direct product \( A^\Omega \) is the set of sequences \( \bar{a} = (a_\omega)_{\omega \in \Omega} \) in \( A \), indexed by \( \Omega \), with a group operation given by pointwise multiplication. The action of \( H \) on \( \Omega \) can be extended to an action on \( A^\Omega \), namely by defining  
    \[
        h \cdot (a_\omega)_{\omega \in \Omega} := (a_{h^{-1}\cdot\omega})_{\omega \in \Omega}
    \]
    for all \( h \in H \) and all \( (a_\omega)_{\omega \in \Omega} \in A^\Omega \). Then the wreath product \( A  \wr_\Omega  H \) of \( A \) by \( H \) is the semidirect product \( A^\Omega \rtimes H \) with the action of \( H \) on \( A^\Omega \) given above. The subgroup \( A^\Omega \) of \( A^\Omega \rtimes H \) is called the base of the wreath product.
\end{definition}

In particular, if \(\Omega = \{0, \cdots, n-1\}\), then we omit \( \Omega \) from the notation above. Thus \( A \wr H \) stands for \( A \wr_{\{0, \cdots, n-1\}} H \), not the regular wreath product notation \( A \wr_\Omega H \).

\section{Permutation Groups of Cyclic Codes with Specific Lengths}\label{Per}

For a prime number $p\neq r$ and a polynomial $g(x)\in\mathbb{F}_{r^\alpha}[x]$ that is a factor of $x^p-1$. The permutation group $\operatorname{Per}(\mathcal{C}_{p,g(x)})$ of $\mathcal{C}_{p,g(x)}$ is fully classified by Lemma \ref{per_cyc_p}. In this section, we determine the permutation groups of cyclic codes with length $hp,r^mp^n,pq$ and generator polynomial $f(x)$ where $f(x)$ and $g(x)$ are related.

\subsection{Cyclic Codes of Length $hp$}

For a positive integer $h$ and a prime number $p$, we have $x^p-1\mid x^{hp}-1$. Therefore, each factor of $x^p-1$ is also a factor of $x^{hp}-1$. In the following, we choose a polynomial $g(x)\in\mathbb{F}_{r^\alpha}[x]$ which is a factor of $x^p-1$ and $\deg (g(x)) = m>1$. Consider two cyclic codes $\mathcal{C}_{p,g(x)}$ and $\mathcal{C}_{hp,g(x)}$. Let $v_0\in\mathcal{C}_{p,g(x)}$ and $v\in\mathcal{C}_{hp,g(x)}$ respectively be the codewords corresponding to the generator polynomial $g(x)$. Using the first type of matrix representation for $\mathcal{C}_{hp,g(x)}$, we have
\[
    Mr_{h}(v) = \begin{pmatrix}
        v_0 \\
        {\bf 0}_p \\
        \vdots \\
        {\bf 0}_p
    \end{pmatrix},
\]
where ${\bf 0}_p$ represents a row vector of $p$-zeros. Let $c\in\mathcal{C}_{hp,g(x)}$ be the generator codeword corresponding to $x^ig(x)$ for some $0\le i\le n-m-1$. Then $Mr_{h}(c)$ can only be in the following two cases
\[
    \begin{pmatrix}
        \cdots \\
        {\bf 0}_p \\
        v_0^{\sigma_0^t} \\
        {\bf 0}_p \\
        \cdots
    \end{pmatrix}\text{ and }\begin{pmatrix}
        \cdots \\
        {\bf 0}_p \\
        v_1 \\
        v_2 \\
        {\bf 0}_p \\
        \cdots
    \end{pmatrix},
\]
where $\sigma_0 = (1,2,\cdots,p)\in S_p$ and $v_1+v_2 = v_0^{\sigma_0^l}$ for some $0\le t,l\le p-1$. 

We are now able to present our result about the permutation group $\operatorname{Per}(\mathcal{C}_{hp,g(x)})$ of $\mathcal{C}_{hp,g(x)}$.

\begin{theorem}\label{L_hp}
    The permutation group $\operatorname{Per}(\mathcal{C}_{hp,g(x)})$ of $\mathcal{C}_{hp,g(x)}$ is isomorphic to $S_h \wr G$, the wreath product of \(S_h\) by \(G\), where $G\cong \operatorname{Per}(\mathcal{C}_{p,g(x)})$, the permutation group of $\mathcal{C}_{p,g(x)}$. 
\end{theorem}
 
\begin{remark}
    Theorem 3.3 in~\cite{maAutomorphismGroupsBinary2025} characterizes a subgroup of the permutation group of binary cyclic code $\mathcal{C}_{hn,g(x)}$ of length $hn$. In this paper we show that, not only over $\mathbb{F}_2$, the permutation group of $\mathcal{C}_{hn,g(x)}$ is indeed isomorphic to the subgroup when $n$ equals a prime number $p$ over \textbf{\textit{arbitrary finite fields }}$\mathbb{F}_{r^\alpha}$.
\end{remark}

Before we provide the proof of Theorem \ref{L_hp}, we first prove the following two lemmas, which consider the permutation group $\operatorname{Per}(\mathcal{C}_{hp,g(x)})$ of the first type of matrix representation for $\mathcal{C}_{hp,g(x)}$ from the perspectives of row permutation of any column and column permutation, respectively. Both two lemmas can be proven by examining the action of specific permutations on the generator codewords and, furthermore, they generalize Lemmas 3.1 and 3.2 in~\cite{maAutomorphismGroupsBinary2025}.

\begin{lemma}\label{L_hp_row}
    The permutation group $\operatorname{Per}(\mathcal{C}_{hp,g(x)})$ of $\mathcal{C}_{hp,g(x)}$ contains a subgroup that is isomorphic to $(S_h)^p$. 
\end{lemma}
\begin{proof}
    Let $c$ be any generator codeword of $\mathcal{C}_{hp,g(x)}$, and $\sigma$ be the cyclic permutation $(b,p+b,\cdots,(h-1)p+b)$ of order $h$, which exactly acts on the $b$-th column of $Mr_{h}(c)$. Then we can derive that 
    \[
        Mr_{h}(c)^\sigma - Mr_{h}(c) = \begin{pmatrix}
            0 & \cdots & 0 & \cdots & 0 \\
              & \ddots &   & \ddots &   \\
            0 & \cdots & -\beta & \cdots & 0 \\
            0 & \cdots & \beta & \cdots & 0 \\
              & \ddots &   & \ddots &   \\
            0 & \cdots & 0 & \cdots & 0 
        \end{pmatrix},
    \]
    where the $b$-th column of $Mr_{h}(c)^\sigma - Mr_{h}(c)$ only has two symbols $-\beta$ and $\beta$ if $\beta\neq 0$ and all other columns are zeros. It is easy to see that $Mr_{h}(c)^\sigma - Mr_{h}(c)$ corresponds to the polynomial $\beta x^{ip+b-1}(x^p-1)$ for some $0\le i\le h-1$. Since $g(x)\mid \beta x^{ip+b-1}(x^p-1)$, we conclude that $Mr_{h}(c)^\sigma - Mr_{h}(c)\in \mathcal{C}_{hp,g(x)}$, hence $Mr_{h}(c)^\sigma\in \mathcal{C}_{hp,g(x)}$. It follows that $\sigma\in \operatorname{Per}(\mathcal{C}_{hp,g(x)})$. 

    Let $\tau$ be the permutation $(b,p+b)$ of order $2$, which also acts on the $b$-th column of $Mr_{h}(c)$. Then, similar to above, we also have $Mr_{h}(c)^\tau - Mr_{h}(c)$ corresponds to the polynomial $\pm\beta(x^p-1)$. Therefore, we also have $\tau\in \operatorname{Per}(\mathcal{C}_{hp,g(x)})$. Consequently, the permutations $\sigma$ and $\tau$ generate a subgroup of $\operatorname{Per}(\mathcal{C}_{hp,g(x)})$ that is isomorphic to $S_h$. Similarly, one can see that $\operatorname{Per}(\mathcal{C}_{hp,g(x)})$ contains a subgroup isomorphic to $(S_h)^p$, which is generate by $\Lambda = \{(b,p+b,\cdots,(h-1)p+b),(b,p+b)\mid 1\le b\le p\}$. The proof is completed.
\end{proof}

\begin{lemma}\label{L_hp_col}
    The permutation group $\operatorname{Per}(\mathcal{C}_{hp,g(x)})$ of $\mathcal{C}_{hp,g(x)}$ contains a subgroup $G$, where $G\cong \operatorname{Per}(\mathcal{C}_{p,g(x)})$, the permutation group of $\mathcal{C}_{p,g(x)}$. 
\end{lemma}
\begin{proof}
    Let $c$ be any generator codeword of $\mathcal{C}_{hp,g(x)}$. For each permutation $\sigma\in S_p$, let $\bar{\sigma}\in S_{hp}$ be the lifting permutation of $\sigma$ such that $\bar{\sigma} = \rho_0(\sigma)\rho_p(\sigma)\rho_{2p}(\sigma)\cdots\rho_{(h-1)p}(\sigma)$ where $\rho_{ip}$ maps coordinates $j$ in $\sigma$ to $ip+j$ with $0\le i \le h-1$ and $1\le j \le p$. Indeed, the permutation $\bar{\sigma}$ is the column permutation of $Mr_{h}(c)$. Now, we claim that $\bar{\sigma}\in \operatorname{Per}(\mathcal{C}_{hp,g(x)})$ if and only if $\sigma\in \operatorname{Per}(\mathcal{C}_{p,g(x)})$. Firstly, it is easy to verify that if
    \[
        Mr_{h}(a^\prime) = \begin{pmatrix}
            \cdots \\
            {\bf 0}_p \\
            a \\
            {\bf 0}_p \\
            \cdots
        \end{pmatrix}
    \]
    for any $a^\prime\in\mathcal{C}_{hp,g(x)}$, then $a\in \mathcal{C}_{p,g(x)}$ by considering the linear combination of the basis of $\mathcal{C}_{hp,g(x)}$. On the contrary, if $a\in \mathcal{C}_{p,g(x)}$, then it is easy to see that
    \[
        \begin{pmatrix}
            \cdots \\
            {\bf 0}_p \\
            a \\
            {\bf 0}_p \\
            \cdots
        \end{pmatrix} = Mr_{h}((0\cdots,0,a,0,\cdots,0))\in\mathcal{C}_{hp,g(x)}.
    \]

    Assume that $\bar{\sigma}\in \operatorname{Per}(\mathcal{C}_{hp,g(x)})$. Then for any $a\in \mathcal{C}_{p,g(x)}$, we have
    \[
        Mr_{h}((a,0,\cdots,0))^{\bar{\sigma}} = \begin{pmatrix}
            a^\sigma \\
            {\bf 0}_p \\
            \cdots \\
            {\bf 0}_p
        \end{pmatrix}\in\mathcal{C}_{hp,g(x)}.
    \]
    Therefore, $a^\sigma\in\mathcal{C}_{p,g(x)}$. It follows that $\sigma\in \operatorname{Per}(\mathcal{C}_{p,g(x)})$. 

    Now, we assume that $\sigma\in \operatorname{Per}(\mathcal{C}_{p,g(x)})$. 
    \begin{itemize}
        \item If $Mr_{h}(c)$ is of the form
        \[
            \begin{pmatrix}
                \cdots \\
                {\bf 0}_p \\
                v_0^{\sigma_0^t} \\
                {\bf 0}_p \\
                \cdots
            \end{pmatrix},
        \]
        then we have
        \[
            Mr_{h}(c)^{\bar{\sigma}} = \begin{pmatrix}
                \cdots \\
                {\bf 0}_p \\
                v_0^{\sigma_0^t\cdot\sigma} \\
                {\bf 0}_p \\
                \cdots
            \end{pmatrix}.
        \]
        Note that $v_0^{\sigma_0^t\cdot\sigma}\in\mathcal{C}_{p,g(x)}$, therefore it can be linearly represented by $v_0,v_0^{\sigma_0},\cdots,v_0^{\sigma_0^{p-m-1}}$. Hence $Mr_{h}(c)^{\bar{\sigma}}\in \mathcal{C}_{hp,g(x)}$. 
        \item If $Mr_{h}(c)$ is of the form
        \[
            \begin{pmatrix}
                \cdots \\
                {\bf 0}_p \\
                v_1 \\
                v_2 \\
                {\bf 0}_p \\
                \cdots
            \end{pmatrix},
        \]
        then we have
        \[
            Mr_{h}(c)^{\bar{\sigma}} = \begin{pmatrix}
                \cdots \\
                {\bf 0}_p \\
                v_1^{\rho_{ip}(\sigma)} \\
                v_2^{\rho_{(i+1)p}(\sigma)} \\
                {\bf 0}_p \\
                \cdots
            \end{pmatrix},
        \]
        for some $0\le i \le h-1$. Since $v_1+v_2 = v_0^{\sigma_0^l}$ for some $p-m\le l\le p-1$, we have $v_1^{\rho_{ip}(\sigma)} + v_2^{\rho_{(i+1)p}(\sigma)} = v_0^{\sigma_0^l\cdot \sigma}\in \mathcal{C}_{p,g(x)}$. Moreover, we have 
        \[
            Mr_{h}(c)^{\bar{\sigma}} - \begin{pmatrix}
                \cdots \\
                {\bf 0}_p \\
                v_1^{\rho_{ip}(\sigma)}+v_2^{\rho_{(i+1)p}(\sigma)} \\
                {\bf 0}_p \\
                {\bf 0}_p \\
                \cdots
            \end{pmatrix} = \begin{pmatrix}
                \cdots \\
                {\bf 0}_p \\
                -v_2^{\rho_{(i+1)p}(\sigma)} \\
                v_2^{\rho_{(i+1)p}(\sigma)} \\
                {\bf 0}_p \\
                \cdots
            \end{pmatrix}\in \mathcal{C}_{hp,g(x)}\,(\text{Similar to the proof of Lemma \ref{L_hp_row}}).
        \]
        Hence, we derive that $Mr_{h}(c)^{\bar{\sigma}}\in \mathcal{C}_{hp,g(x)}$.
    \end{itemize}
    From the above deduction, we conclude that $\bar{\sigma}\in \operatorname{Per}(\mathcal{C}_{hp,g(x)})$. It follows that all permutation $\bar{\sigma}\in\operatorname{Per}(\mathcal{C}_{hp,g(x)})$ that permute columns generate a subgroup $G$ which is isomorphic to $\operatorname{Per}(\mathcal{C}_{p,g(x)})$. The proof is completed.
\end{proof}

\noindent{\raggedright{\it Proof of Theorem \ref{L_hp}.}} Note that by Lemmas \ref{L_hp_row} and \ref{L_hp_col}, the permutation group $G$, where $G\cong \operatorname{Per}(\mathcal{C}_{p,g(x)})$, normalizes but not commutes with $(S_h)^p$. Therefore, we can see that $\operatorname{Per}(\mathcal{C}_{hp,g(x)})$ contains a subgroup that is isomorphic to $S_h \wr G$, the wreath product of \(S_h\) by \(G\). In the following, we will prove that the remaining permutations belong to the set $S_{hp}\backslash(S_h \wr G)$ are not in $\operatorname{Per}(\mathcal{C}_{hp,g(x)})$. 

Let $\varphi\in S_{hp}\backslash(S_h \wr G)$. By multiplying $\varphi$ with a suitable permutation of $S_h \wr G$, we can obtain a permutation that fixes the $(1,1)$-entry and maps the $(2,1)$-entry to the $(1,b)$-entry of each matrix representation of codeword for some $1< b\le p$. Hence, we may assume that $\varphi$ fixes the $(1,1)$-entry and maps the $(2,1)$-entry to the $(1,b)$-entry. Let $c$ be the codeword corresponding to the polynomial $x^p-1$. Then we have 
\[
    Mr_{h}(c)^\varphi = \begin{pmatrix}
        -1 & 0 & \cdots & 0 \\
        1 & 0 & \cdots & 0 \\
        \vdots & \vdots & \ddots & \vdots \\
        0 & 0 & \cdots & 0
    \end{pmatrix}^\varphi = \begin{pmatrix}
        -1 & \cdots & 1 & \cdots & 0 \\
        0 & \cdots & 0 & \cdots & 0 \\
        \vdots & \ddots & \vdots & \ddots & \vdots \\
        0 & \cdots & 0 & \cdots & 0
    \end{pmatrix},
\]
which corresponds to the polynomial $x^{b-1}-1$. If $Mr_{h}(c)^\varphi\in \mathcal{C}_{hp,g(x)}$, then we have 
\[
    g(x)\mid \gcd(x^{b-1}-1,x^p-1) = x-1,
\]
which is a contradiction since $\deg(g(x))>1$. Therefore, $Mr_{h}(c)^\varphi\notin \mathcal{C}_{hp,g(x)}$. It follows that $\varphi\notin \operatorname{Per}(\mathcal{C}_{hp,g(x)})$. Hence $\operatorname{Per}(\mathcal{C}_{hp,g(x)})\cong S_h \wr G$. The proof is completed. \qed

\subsection{Cyclic Codes of Length $r^mp^n$}

Let $n$ be a positive number and $m$ be a non-negative number. In the following, we choose a polynomial $g(x)\in\mathbb{F}_{r^\alpha}[x]$ which is a factor of $x^p-1$, then we can derive that
\[
    g(x^{r^up^v})\,\big|\, x^{r^up^{v+1}}-1\,\big|\, x^{r^mp^n} - 1,
\]
for any $0\le u\le m$ and $0\le v\le n-1$. Therefore $g(x^{r^up^v})$ is indeed a factor of $x^{r^mp^n} - 1$. Consider two cyclic codes $\mathcal{C}_{r^{m-u}p^{n-v},g(x)}$ and $\mathcal{C}_{r^mp^n,g(x^{r^up^v})}$. Let $v_0\in\mathcal{C}_{r^{m-u}p^{n-v},g(x)}$ and $v\in\mathcal{C}_{r^mp^n,g(x^{r^up^v})}$ be the codewords corresponding to the generator polynomial $g(x)$ and $g(x^{r^up^v})$ respectively. Then we use the second type of matrix representation for $\mathcal{C}_{r^mp^n,g(x^{r^up^v})}$, we have
\[
    Mc_{r^up^v}(v) = \begin{pmatrix}
        v_0 \\
        {\bf 0} \\
        \cdots \\
        {\bf 0}
    \end{pmatrix},
\]
where ${\bf 0}$ represents a row vector of $r^{m-u}p^{n-v}$-zeros. 

Now we present our result about the permutation group of $\mathcal{C}_{r^mp^n,g(x^{r^up^v})}$.

\begin{theorem}\label{L_rp}
    For any $0\le u\le m$ and $0\le v\le n-1$. The permutation group $\operatorname{Per}(\mathcal{C}_{r^mp^n,g(x^{r^up^v})})$ of $\mathcal{C}_{r^mp^n,g(x^{r^up^v})}$ is isomorphic to $G \wr S_{r^up^v}$, the wreath product of \(G\) by \(S_{r^up^v}\), where $G\cong \operatorname{Per}(\mathcal{C}_{r^{m-u}p^{n-v},g(x)})$, the permutation group of $\mathcal{C}_{r^{m-u}p^{n-v},g(x)}$. 
\end{theorem}

\begin{remark}
    The above result extends those in~\cite{maAutomorphismGroupsBinary2025} by determining the permutation groups of cyclic codes with more lengths and by loosing the restriction that the generator polynomials are irreducible.
\end{remark}

\noindent{\raggedright{\it Proof of Theorem \ref{L_rp}.}} Let $\bar{\sigma} = (1,2,\cdots,r^mp^n)\in \operatorname{Per}(\mathcal{C}_{r^mp^n,g(x^{r^up^v})})$. Then we have
    \[
        Mc_{r^up^v}(v)^{\bar{\sigma}^L} = \begin{pmatrix}
            {\bf 0} \\
            \cdots \\
            v_0^{\sigma^l} \\
            \cdots \\
            {\bf 0}
        \end{pmatrix},
    \]
    where $\sigma = (1,2,\cdots,r^{m-u}p^{n-v})\in\operatorname{Per}(\mathcal{C}_{r^{m-u}p^{n-v},g(x)})$, $L = l\cdot r^up^v + t$ with $0\le t< r^up^v$ and $v_0^{\sigma^l}$ lies in the $t+1$-row. Additionally, for each codeword $w\in\mathcal{C}_{r^mp^n,g(x^{r^up^v})}$ and
    \[
        Mc_{r^up^v}(w) = \begin{pmatrix}
            w_1 \\
            w_2 \\
            \cdots \\
            w_{r^up^v}
        \end{pmatrix},
    \]
    we can derive that $w_i\in\mathcal{C}_{r^{m-u}p^{n-v},g(x)}$ with $1\le i\le r^up^v$ by considering the linear combination of the basis of $\mathcal{C}_{r^mp^n,g(x^{r^up^v})}$. The converse is also true. Therefore, we can see that any permutation that acts on $i$-th row of $Mc_{r^up^v}(w)$ belongs to $\operatorname{Per}(\mathcal{C}_{r^mp^n,g(x^{r^up^v})})$ if and only if its restriction on $w_i$ belongs to $\operatorname{Per}(\mathcal{C}_{r^{m-u}p^{n-v},g(x)})$. Consequently, these permutations generate a subgroup isomorphic to $G^{r^up^v}$, where $G\cong \operatorname{Per}(\mathcal{C}_{r^{m-u}p^{n-v},g(x)})$. Furthermore, any permutation that is a row permutation of $Mc_{r^up^v}(w)$ clearly belongs to $\operatorname{Per}(\mathcal{C}_{r^mp^n,g(x^{r^up^v})})$. Similar to the proof of Lemma \ref{L_hp_row}, we know that these permutations generate a subgroup isomorphic to $S_{r^up^v}$, which normalizes but not commutes with $G^{r^up^v}$. Hence, we can see that $\operatorname{Per}(\mathcal{C}_{r^mp^n,g(x^{r^up^v})})$ contains a subgroup that is isomorphic to $G \wr S_{r^up^v}$, the wreath product of \(G\) by \(S_{r^up^v}\). In the following, we will prove that the remaining permutations belong to the set $S_{r^mp^n}\backslash(G \wr S_{r^up^v})$ are not in $\operatorname{Per}(\mathcal{C}_{r^mp^n,g(x^{r^up^v})})$. 

    Let $\varphi\in S_{r^mp^n}\backslash(G \wr S_{r^up^v})$. By multiplying $\varphi$ with a suitable permutation of $G \wr S_{r^up^v}$, we can assume that $\varphi$ maintains the $(1, 1)$-entry in the first row and maps certain entries in the first row of $Mc_{r^up^v}(w)$ to other rows. Next, we firstly prove the following claim.

    \noindent{\bf Claim.} If $r^{m-u}p^{n-v}>p$ and $\varphi\in \operatorname{Per}(\mathcal{C}_{r^mp^n,g(x^{r^up^v})})$, then, for any $1\le b\le p$, $\varphi$ maps $(b,1),(p+b,1),\cdots,((r^{m-u}p^{n-v-1}-1)p+b,1)$ entries to the same row.

    \noindent{\raggedright{\it Proof of Claim.}} If not, then W.O.L.G. we can assume that $\varphi$ maps $(b,1)$ and $(p+b,1)$ entries to different rows. Consider $w\in\mathcal{C}_{r^mp^n,g(x^{r^up^v})}$ with
    \[
        Mc_{r^up^v}(w) = \begin{pmatrix}
            \cdots & -1 & \cdots & 1 & \cdots \\
            \cdots & 0 & \cdots & 0 & \cdots \\
             &   & \cdots &  & 
        \end{pmatrix},
    \]
    where the $(b,1)$-entry is $-1$, the $(p+b,1)$-entry is $1$ and other entries are $0$. Then we have
    \[
        Mc_{r^up^v}(w)^{\varphi} = \begin{pmatrix}
             &  & \cdots &  &  \\
            \cdots & -1 & \cdots & 0 & \cdots \\
             &  & \cdots &  &  \\
            \cdots & 0 & \cdots & 1 & \cdots \\
             &  & \cdots &  & 
        \end{pmatrix}\notin \mathcal{C}_{r^mp^n,g(x^{r^up^v})}, 
    \]
    since $(\cdots,0,1,0,\cdots)\notin\mathcal{C}_{r^{m-u}p^{n-v},g(x)}$, which contradicts with $\varphi\in \operatorname{Per}(\mathcal{C}_{r^mp^n,g(x^{r^up^v})})$. Hence, the claim is right. 

    By the above claim, we can assume that $\varphi$ maintains the $(1, 1)$-entry in the first row and maps certain entries in $\{(i,1)\mid 2\le i \le p\}$ to other rows. For any $a = (1,\cdots,\beta,\cdots)\in\mathcal{C}_{p,g(x)}$, consider $a^\prime\in\mathcal{C}_{r^mp^n,g(x^{r^up^v})}$ with
    \[
        Mc_{r^up^v}(a^\prime) = \begin{pmatrix}
            a & 0 & \cdots  \\
            0 & 0 & \cdots \\
              & \cdots & 
        \end{pmatrix}. 
    \]
    Suppose $\varphi$ maintains all non-zero entries in $a$ in the first row and $\beta$ lies in the $(b_0+1,1)$-entry. Let $\tau = (1,2,\cdots,p)\in \operatorname{Per}(\mathcal{C}_{p,g(x)})$ and $a^\prime_t\in\mathcal{C}_{r^mp^n,g(x^{r^up^v})}$ with 
    \[
        Mc_{r^up^v}(a^\prime_t) = \begin{pmatrix}
            a^{\tau^t} & 0 & \cdots  \\
            0 & 0 & \cdots \\
              & \cdots & 
        \end{pmatrix},
    \]
    for $0\le t\le p-1$. Since $\gcd(b_0,p) = 1$, we can derive that
    \[
        \{tb_0\bmod{p}\mid 0\le t\le p-1\} = \{0,1,\cdots,p-1\}. 
    \]
    Therefore, there must exist $t_0\in\{0,1,\cdots,p-1\}$ such that $\varphi$ maintains certain non-zero entries in $a^{\tau^{t_0b_0}}$ in the first row and maps certain non-zero entries in $a^{\tau^{t_0b_0}}$ to other rows of $Mc_{r^up^v}(a^\prime_{t_0b_0})$. Otherwise, all entries in $\{(i,1)\mid 1\le i \le p\}$ would be maintained in the first row by $\varphi$ which contradicts with the assumption of $\varphi$. So, by multiplying $\varphi$ with a suitable permutation, we can assume that $\varphi$ fixes the $(1, 1)$-entry and maps certain non-zero entries in $a$ to other rows. 

    Suppose that the minimum distance of $\mathcal{C}_{p,g(x)}$ is $d$ and $\varphi\in \operatorname{Per}(\mathcal{C}_{r^mp^n,g(x^{r^up^v})})$. We choose $a = (1,\cdots)\in\mathcal{C}_{p,g(x)}$ such that the Hamming weight of $a$ equals $d$. Let $a^\prime\in\mathcal{C}_{r^mp^n,g(x^{r^up^v})}$ with
        \[
            Mc_{r^up^v}(a^\prime) = \begin{pmatrix}
                a & 0 & \cdots  \\
                0 & 0 & \cdots \\
                  & \cdots & 
            \end{pmatrix}. 
        \]
        By the assumption of $\varphi$, we have
        \[
            Mc_{r^up^v}(a^\prime)^\varphi = \begin{pmatrix}
                a_1 \\
                a_2 \\
                \cdots \\
                a_{r^u}
            \end{pmatrix},
        \]
        and there must exist $i\in \{2,\cdots,p\}$ such that $\varphi$ maps the $(i,1)$ non-zero entry in $a$ to the $(jp+1,1)$-entry of $Mc_{r^up^v}(a^\prime)^\varphi$ for some $j\in\{1,\cdots,r^{m-u}p^{n-v-1}-1\}$. Otherwise, by adding a suitable codeword $c\in\mathcal{C}_{r^{m-u}p^{n-v},g(x)}$ corresponding to the polynomial $\tilde{g}(x)(x^p-1)$ for some $\tilde{g}(x)\in \mathbb{F}_{r^\alpha}[x]$, we have
        \[
            a_1 + c = (a^\prime{}^\prime,0,\cdots,0),
        \]
        where ${\bf 0}\neq a^\prime{}^\prime \in \mathcal{C}_{p,g(x)}$ with Hamming weight less than $d$, which is a contradiction. Now, we consider $w_t\in\mathcal{C}_{r^mp^n,g(x^{r^up^v})}$ corresponding to the polynomial $x^{tp}-1$ for all $t\in\{1,\cdots,r^{m-u}p^{n-v-1}-1\}$ with
        \[
            Mc_{r^up^v}(w_t) = \begin{pmatrix}
                -1 & \cdots & 1 & \cdots \\
                0 & \cdots & 0 & \cdots \\
                  & \cdots &  & 
            \end{pmatrix}.
        \]
        Similarly, by the assumption of $\varphi$, the permutation $\varphi$ must map the $(tp + 1,1)$-entry of $Mc_{r^up^v}(w_t)$ to the $(t^\prime p+1,1)$-entry of $Mc_{r^up^v}(w_t)^\varphi$ for some $t^\prime\in\{1,\cdots,r^{m-u}p^{n-v-1}-1\}$. Therefore, there exists $t_0\in\{1,\cdots,r^{m-u}p^{n-v-1}-1\}$ such that $\varphi$ maps the $(t_0p + 1,1)$-entry to the $(jp+1,1)$-entry, which is a contradiction since $\varphi$ also maps the $(i,1)$-entry to the $(jp+1,1)$-entry. Hence, we derive that $\varphi\notin \operatorname{Per}(\mathcal{C}_{r^mp^n,g(x^{r^up^v})})$. From the above deduction, we conclude that $\operatorname{Per}(\mathcal{C}_{r^mp^n,g(x^{r^up^v})})\cong G \wr S_{r^up^v}$. The proof is completed. \qed

Note that $x^{r^mp^n}-1\mid x^{hr^mp^n}-1$ for any positive integer $h$. Similar to the proof of Theorem~\ref{L_rp}, we can obtain the following corollary.

\begin{corollary}\label{L_rp_col}
    For any $0\le u\le m$ and $0\le v\le n-1$. The permutation group $\operatorname{Per}(\mathcal{C}_{hr^mp^n,g(x^{r^up^v})})$ of $\mathcal{C}_{hr^mp^n,g(x^{r^up^v})}$ is isomorphic to $G \wr S_{r^up^v}$, the wreath product of \(G\) by \(S_{r^up^v}\), where $G\cong \operatorname{Per}(\mathcal{C}_{hr^{m-u}p^{n-v},g(x)})$, the permutation group of $\mathcal{C}_{hr^{m-u}p^{n-v},g(x)}$. 
\end{corollary}

\subsection{Cyclic Codes of Length $pq$}

For prime numbers $p\neq r$ and $q\neq r$, we have
\[
    x^{pq} - 1 = (x-1)\cdot Q_p(x)\cdot Q_q(x)\cdot Q_{pq}(x) 
\]
where $Q_{n}(x)$ is well known as the $n$-th cyclotomic polynomial over $\mathbb{F}_{r^\alpha}$ for any positive integer $n$ \cite[Chap. 2]{Lidl_Niederreiter_1996}. Consider the cyclic codes $\mathcal{C}_{pq,Q_{pq}(x)}$ and $\mathcal{C}_{pq,(x-1)Q_p(x)Q_q(x)}$. By the properties of cyclic codes, we can easily obtain that 
\[
    \mathcal{C}_{pq,Q_{pq}(x)}^\perp = \mathcal{C}_{pq,(x-1)Q_p(x)Q_q(x)}. 
\]
From the polynomial correspondence of cyclic codes in (\ref{cyc_poly_cor}), we can derive that 
\[
    \mathcal{C}_{pq,(x-1)Q_p(x)Q_q(x)} = \mathcal{C}_{pq,x-1}\cap\mathcal{C}_{pq,Q_p(x)}\cap\mathcal{C}_{pq,Q_q(x)}. 
\]
By Theorem \ref{L_hp}, we already know that the permutation groups of $\mathcal{C}_{pq,Q_p(x)}$ and $\mathcal{C}_{pq,Q_q(x)}$. Additionally, it is easily to obtain that $\operatorname{Per}(\mathcal{C}_{pq,x-1}) = S_{pq}$. Therefore, we want to investigate the relationship between $\operatorname{Per}(\mathcal{C}_{pq,(x-1)Q_p(x)Q_q(x)})$ and $\operatorname{Per}(\mathcal{C}_{pq,Q_p(x)})(\text{or }\operatorname{Per}(\mathcal{C}_{pq,Q_q(x)}))$ naturally. Hence, in the following, we present our main result about the permutation groups $\operatorname{Per}(\mathcal{C}_{pq,Q_{pq}(x)})$ and $\operatorname{Per}(\mathcal{C}_{pq,(x-1)Q_p(x)Q_q(x)})$.

\begin{theorem}\label{L_pq}
    The permutation group $\operatorname{Per}(\mathcal{C}_{pq,Q_{pq}(x)})$ of $\mathcal{C}_{pq,Q_{pq}(x)}$ is isomorphic to $S_p\times S_q$. Moreover, we have
    \begin{align*}
        \operatorname{Per}(\mathcal{C}_{pq,Q_{pq}(x)}) &= \operatorname{Per}(\mathcal{C}_{pq,(x-1)Q_p(x)Q_q(x)}) \\
        &= \operatorname{Per}(\mathcal{C}_{pq,Q_p(x)})\cap \operatorname{Per}(\mathcal{C}_{pq,Q_q(x)}) \\
        &\cong S_p\times S_q. 
    \end{align*}
\end{theorem}
 
\begin{remark}
    We note that the above result provides what we believe to be the first determination of the permutation groups of cyclic codes whose generator polynomials are factors of $x^{pq}-1$ but not factors of $x^p-1(\text{or }x^q-1)$.
\end{remark}

Before we provide the proof of Theorem \ref{L_pq}, it is essential to introduce the following lemmas concerning the permutation groups of the dual cyclic codes and certain properties associated with the inclusion relationship of permutation groups of cyclic codes.

\begin{lemma}\label{dual_per}
    For a cyclic code $\mathcal{C}$, we have $\operatorname{Per}(\mathcal{C}) = \operatorname{Per}(\mathcal{C}^\perp)$. 
\end{lemma}
\begin{proof}
    For any $c\in\mathcal{C}$, $d\in\mathcal{C}^\perp$ and $\sigma \in \operatorname{Per}(\mathcal{C})$, we have
    \[
        \langle c,d \rangle = \langle c^\sigma,d^\sigma\rangle = \langle c^\prime,d^\sigma\rangle = 0,
    \]
    where $c^\prime$ ranges over all elements in $\mathcal{C}$. Therefore, we have $d^\sigma\in\mathcal{C}^\perp$. It follows that $\sigma\in\operatorname{Per}(\mathcal{C}^\perp)$. For any $\tau \in \operatorname{Per}(\mathcal{C}^\perp)$, the above derivation also holds. Hence, we conclude that $\operatorname{Per}(\mathcal{C}) = \operatorname{Per}(\mathcal{C}^\perp)$. The proof is completed. 
\end{proof}

By Lemma \ref{dual_per}, we can thus obtain $\operatorname{Per}(\mathcal{C}_{pq,Q_{pq}(x)}) = \operatorname{Per}(\mathcal{C}_{pq,(x-1)Q_p(x)Q_q(x)})$. We now turn our attention exclusively to $\operatorname{Per}(\mathcal{C}_{pq,(x-1)Q_p(x)Q_q(x)})$.

\begin{lemma}\label{inclu_per}
    For $\mathcal{C}_{pq,(x-1)Q_p(x)Q_q(x)}$ and $\mathcal{C}_{pq,Q_p(x)}(\text{resp. }\mathcal{C}_{pq,Q_q(x)})$, we have
    \[
        \operatorname{Per}(\mathcal{C}_{pq,(x-1)Q_p(x)Q_q(x)})\leqslant \operatorname{Per}(\mathcal{C}_{pq,Q_p(x)})\left(\text{resp. }\operatorname{Per}(\mathcal{C}_{pq,Q_q(x)})\right). 
    \]
\end{lemma}
\begin{proof}
    W.L.O.G. we can assume that $p<q$. For $\mathcal{C}_{pq,Q_q(x)}$, by Theorem \ref{L_hp}, we have $\operatorname{Per}(\mathcal{C}_{pq,Q_q(x)})\cong S_p \wr S_q$. Let $\varphi\notin \operatorname{Per}(\mathcal{C}_{pq,Q_q(x)})$. We use the first type of matrix representation for $\mathcal{C}_{pq,(x-1)Q_p(x)Q_q(x)}$, and consider the codeword $w_i\in\mathcal{C}_{pq,(x-1)Q_p(x)Q_q(x)}$ corresponding to the polynomial $(x^p-1)(x^{iq}-1)$ for $1\le i\le p-1$ with
    \[
        Mr_p(w_i) = \begin{pmatrix}
            1 & \cdots & -1 & \cdots \\
              & \cdots &  & \\
            -1 & \cdots & 1 & \cdots \\
              & \cdots &  & 
        \end{pmatrix},
    \]
    where the $(1,p+1)$-entry is $-1$, the $(i+1,1)$-entry is $-1$ and the $(i+1,p+1)$-entry is $1$. We can assume that $\varphi$ fixes the $(1,1)$-entry and maps $(1,p+1)$ and $(i_0+1,p+1)$ entries to two different columns for some $i_0\in \{1,\cdots,p-1\}$. 
    \begin{enumerate}
        \item If $Mr_p(w_{i_0})^\varphi$(act by a suitable permutation of $\operatorname{Per}(\mathcal{C}_{pq,Q_q(x)})$ if necessary) equals one of the following four cases
        \[
            \begin{pmatrix}
                1 & -1 & \cdots \\
                1 & -1 & \cdots \\
                  &  & \cdots 
            \end{pmatrix},\,\begin{pmatrix}
                1 & 1 & -1 & \cdots \\
                -1 & & \cdots & \\
                  &  & \cdots & 
            \end{pmatrix},
        \]
        \[
            \begin{pmatrix}
                1 & -1 & \cdots \\
                -1 & & \cdots \\
                1  & & \cdots \\
                  & & \cdots
            \end{pmatrix},\,\begin{pmatrix}
                1 & -1 & -1 & 1 & \cdots \\
                  & & \cdots  & 
            \end{pmatrix}, 
        \]
        then we can easily derive that the above four codewords are not in $\mathcal{C}_{pq,Q_q(x)}$, and of course not in $\mathcal{C}_{pq,(x-1)Q_p(x)Q_q(x)}$. Therefore, we have $\varphi\notin \operatorname{Per}(\mathcal{C}_{pq,(x-1)Q_p(x)Q_q(x)})$. 
        
        \item If 
        \[
            Mr_p(w_{i_0})^\varphi= \begin{pmatrix}
                1 & \cdots & -1 & \cdots \\
                  & \cdots &  & \\
                -1 & \cdots & 1 & \cdots \\
                  & \cdots &  & 
            \end{pmatrix},
        \]
        which is still in $\mathcal{C}_{pq,(x-1)Q_p(x)Q_q(x)}$, then we can consider $Mr_p(w_{i_0})^{\tau^{p}}$ with $\tau = (1,2,\cdots,pq)$ and keep go on. It can be seen that if, for all $0\le t\le q-1$, we have $Mr_p(w_{i_0})^{\tau^{tp}\cdot\varphi}\in \mathcal{C}_{pq,(x-1)Q_p(x)Q_q(x)}$, then $\varphi$ must map the $(\ell_1+1,\ell_2)$-entry to the first column where $(t+1)p+1 = \ell_1q+\ell_2$ with $0\le \ell_2<q$, which is a contradiction since this will result in $q$ positions being mapped to the first column but $p<q$. Therefore, there exists $t_0\in \{0,\cdots,q-1\}$ such that $Mr_p(w_{i_0})^{\tau^{t_0p}\cdot\varphi}\notin \mathcal{C}_{pq,(x-1)Q_p(x)Q_q(x)}$. It follows that $\varphi\notin \operatorname{Per}(\mathcal{C}_{pq,(x-1)Q_p(x)Q_q(x)})$. 
    \end{enumerate}
    From the above deduction, we conclude that $\operatorname{Per}(\mathcal{C}_{pq,(x-1)Q_p(x)Q_q(x)})\leqslant \operatorname{Per}(\mathcal{C}_{pq,Q_q(x)})$. 

    For $\mathcal{C}_{pq,Q_p(x)}$, we also use the first type of matrix representation for $\mathcal{C}_{pq,(x-1)Q_p(x)Q_q(x)}$, and consider $w_i\in\mathcal{C}_{pq,(x-1)Q_p(x)Q_q(x)}$ corresponding to $(x^p-1)(x^{iq}-1)$ for $1\le i\le p-1$ with
    \[
        Mr_q(w_i) = \begin{pmatrix}
            1 & \cdots &  &  \\
            -1 & \cdots &  & \\
              & \cdots & -1 & \cdots \\
              & \cdots & 1 & \cdots \\
              & \cdots &  &
        \end{pmatrix},
    \]
    where the $(2,1)$-entry is $-1$, the $(\ell_{i_1}+1,\ell_{i_2})$-entry is $-1$ and the $(\ell_{i_1}+2\bmod{q},\ell_{i_2})$-entry is $1$ with $iq+1 = \ell_{i_1}p+\ell_{i_2}$ and $0\le\ell_2<p$. Let $\varphi\notin \operatorname{Per}(\mathcal{C}_{pq,Q_p(x)})$. We can assume that $\varphi$ fixes the $(1,1)$-entry and maps $(\ell_{i^\prime_1}+1,\ell_{i^\prime_2})$ and $(\ell_{i^\prime_1}+2\bmod{q},\ell_{i^\prime_2})$ entries to two different columns for some $i^\prime\in \{1,\cdots,p-1\}$. Then, similar to the above derivation, if $Mr_q(w_{i^\prime})^\varphi$ is still in $\mathcal{C}_{pq,(x-1)Q_p(x)Q_q(x)}$, then by considering $Mr_q(w_{i^\prime})^{\tau^{\ell q+tp}\cdot\varphi}$ with $0\le \ell\le p-1$ and $0\le t\le q-1$, we can derive that the entries corresponding to $x^{tp},x^{q+tp},\cdots,x^{(p-1)q+tp}$ are mapped to the same column. Therefore, there are $\ell^\prime p$ positions that are mapped to the first column where $\ell^\prime$ is a positive integer. It is a contradiction since $\gcd(p,q)=1$. Hence, we have $\varphi\notin \operatorname{Per}(\mathcal{C}_{pq,(x-1)Q_p(x)Q_q(x)})$, and thus we conclude that $\operatorname{Per}(\mathcal{C}_{pq,(x-1)Q_p(x)Q_q(x)})\leqslant \operatorname{Per}(\mathcal{C}_{pq,Q_p(x)})$. The proof is completed. 
\end{proof}

\noindent{\raggedright{\it Proof of Theorem \ref{L_pq}.}} Firstly, by Lemma \ref{inclu_per}, we derive that
\[
    \operatorname{Per}(\mathcal{C}_{pq,(x-1)Q_p(x)Q_q(x)})\leqslant \operatorname{Per}(\mathcal{C}_{pq,Q_p(x)})\cap\operatorname{Per}(\mathcal{C}_{pq,Q_q(x)}). 
\]
Secondly, for any $\varphi\in \operatorname{Per}(\mathcal{C}_{pq,x-1})\cap\operatorname{Per}(\mathcal{C}_{pq,Q_p(x)})\cap\operatorname{Per}(\mathcal{C}_{pq,Q_q(x)})$ and $w\in \mathcal{C}_{pq,(x-1)Q_p(x)Q_q(x)}$, we have
\[
    w^\varphi\in \mathcal{C}_{pq,x-1}\cap\mathcal{C}_{pq,Q_p(x)}\cap\mathcal{C}_{pq,Q_q(x)} = \mathcal{C}_{pq,(x-1)Q_p(x)Q_q(x)}.
\]
It follows that $\varphi \in \operatorname{Per}(\mathcal{C}_{pq,(x-1)Q_p(x)Q_q(x)})$. Therefore, we have
\begin{align*}
    \operatorname{Per}(\mathcal{C}_{pq,x-1})\cap\operatorname{Per}(\mathcal{C}_{pq,Q_p(x)})\cap\operatorname{Per}(\mathcal{C}_{pq,Q_q(x)}) &= \operatorname{Per}(\mathcal{C}_{pq,Q_p(x)})\cap\operatorname{Per}(\mathcal{C}_{pq,Q_q(x)}) \\
    &\leqslant \operatorname{Per}(\mathcal{C}_{pq,(x-1)Q_p(x)Q_q(x)}).
\end{align*}
And thus, by Lemma \ref{dual_per}, we can conclude that
\begin{align*}
    \operatorname{Per}(\mathcal{C}_{pq,Q_{pq}(x)}) &= \operatorname{Per}(\mathcal{C}_{pq,(x-1)Q_p(x)Q_q(x)}) \\
        &= \operatorname{Per}(\mathcal{C}_{pq,Q_p(x)})\cap \operatorname{Per}(\mathcal{C}_{pq,Q_q(x)}) \\
        &\cong (S_q \wr S_p)\cap(S_p \wr S_q) \\
        &\cong S_p\times S_q. 
\end{align*}
The proof is completed. \qed

From the proof of Lemma~\ref{inclu_per}, we can also obtain the following corollary.

\begin{corollary}\label{L_pq_col}
    For any positive integer $h$, we have
    \[
        \operatorname{Per}(\mathcal{C}_{hpq,Q_p(x)Q_q(x)}) = \operatorname{Per}(\mathcal{C}_{hpq,Q_p(x)})\cap\operatorname{Per}(\mathcal{C}_{hpq,Q_q(x)})\cong S_h\wr(S_p\times S_q).
    \]
\end{corollary}
\begin{proof}
    Firstly, similar to the proof of Lemma~\ref{inclu_per}, we also have
    \[
        \operatorname{Per}(\mathcal{C}_{hpq,Q_p(x)Q_q(x)})\leqslant \operatorname{Per}(\mathcal{C}_{hpq,Q_p(x)}),
    \]
    and 
    \[
        \operatorname{Per}(\mathcal{C}_{hpq,Q_p(x)Q_q(x)})\leqslant \operatorname{Per}(\mathcal{C}_{hpq,Q_q(x)}). 
    \]
    Therefore, similar to the proof of Theorem~\ref{L_pq}, we can derive that
    \begin{align*}
       \operatorname{Per}(\mathcal{C}_{hpq,Q_p(x)Q_q(x)}) &= \operatorname{Per}(\mathcal{C}_{hpq,Q_p(x)})\cap \operatorname{Per}(\mathcal{C}_{hpq,Q_q(x)}) \\
        &\cong (S_{hq} \wr S_p)\cap(S_{hp} \wr S_q) \\
        &\cong S_h \wr (S_p\times S_q). 
    \end{align*}
    The proof is completed. 
\end{proof}

\section{Conclusion}\label{Con}

\begin{table*}[!t]
    \caption{Binary Cyclic Codes $\mathcal{C}_{n,g(x)}$ of Length $n$ with Generator Polynomial $g(x)$}
	\centering
    \renewcommand{\arraystretch}{1.5}
	\begin{tabular}{lll}
        \toprule[1.2pt]
		Length $n$ & Generator polynomial $g(x)$ & Permutation group ${\rm Per}(\mathcal{C}_{n,g(x)})$ \\
        \midrule[1.2pt]
		$7$ & $x^3+x+1\text{ or }x^3+x^2+1$ & $\operatorname{PSL}(2,7)$ \\   
        
		$2\cdot 7$ & $x^3+x+1\text{ or }x^3+x^2+1$ & $S_2 \wr \operatorname{PSL}(2,7)$ \\   
        
		$2\cdot 7$ & $(x^3+x+1)^2\text{ or }(x^3+x^2+1)^2$ & $\operatorname{PSL}(2,7) \wr S_2$ \\   

        $3\cdot 7$ & $x^3+x+1\text{ or }x^3+x^2+1$ & $S_3 \wr \operatorname{PSL}(2,7)$ \\   

        $3\cdot 2\cdot 7$ & $x^3+x+1\text{ or }x^3+x^2+1$ & $S_{6} \wr \operatorname{PSL}(2,7) $ \\   
        
		$7^2$ & $x^3+x+1\text{ or }x^3+x^2+1$ & $S_7 \wr \operatorname{PSL}(2,7)$ \\   
        
		$7^2$ & $x^{21}+x^7+1\text{ or }x^{21}+x^{14}+1$ & $\operatorname{PSL}(2,7) \wr S_7$ \\   
        
        $2\cdot 7^2$ & $x^{21}+x^7+1\text{ or }x^{21}+x^{14}+1$ & $\big(S_2 \wr \operatorname{PSL}(2,7)\big) \wr S_7$ \\   
        
		$2\cdot 7^2$ & $(x^{21}+x^7+1)^2\text{ or }(x^{21}+x^{14}+1)^2$ & $\operatorname{PSL}(2,7) \wr S_{14}$ \\   

        $2^2\cdot 7^2$ & $(x^3+x+1)^4\text{ or }(x^3+x^2+1)^4$ & $\big(S_7 \wr \operatorname{PSL}(2,7)\big) \wr S_4$ \\   
        
		$2^2\cdot 7^2$ & $(x^{21}+x^7+1)^2\text{ or }(x^{21}+x^{14}+1)^2$ & $\big(S_2 \wr \operatorname{PSL}(2,7)\big) \wr S_{14}$ \\   

        $2^2\cdot 7^2$ & $(x^{21}+x^7+1)^4\text{ or }(x^{21}+x^{14}+1)^4$ & $ \operatorname{PSL}(2,7) \wr S_{28}$ \\   

        $3\cdot2\cdot 7^2$ & $x^{21}+x^7+1\text{ or }x^{21}+x^{14}+1$ & $\big(S_6 \wr \operatorname{PSL}(2,7)\big) \wr S_{7}$ \\   

        $3\cdot2\cdot 7^2$ & $(x^{21}+x^7+1)^2\text{ or }(x^{21}+x^{14}+1)^2$ & $\big(S_3 \wr \operatorname{PSL}(2,7)\big) \wr S_{14}$ \\   

        $31$ & $(x^5 + x^2 + 1)(x^5 + x^3 + 1)(x^5 + x^3 + x^2 + x + 1)$ & $C_{31}\rtimes C_5$ \\   

        $2\cdot 31$ & $(x^5 + x^2 + 1)^2(x^5 + x^3 + 1)^2(x^5 + x^3 + x^2 + x + 1)^2$ & $(C_{31}\rtimes C_5) \wr S_2$ \\   

        $31^2$ & $(x^{155} + x^{62} + 1)(x^{155} + x^{93} + 1)(x^{155} + x^{93} + x^{62} + x^{31} + 1)$ & $(C_{31}\rtimes C_5) \wr S_{31}$ \\   

        $2^2\cdot 31^2$ & $(x^{155} + x^{62} + 1)^2(x^{155} + x^{93} + 1)^2(x^{155} + x^{93} + x^{62} + x^{31} + 1)^2$ & $\big(S_{2} \wr (C_{31}\rtimes C_5)\big) \wr S_{62}$ \\   

        $3$ & $Q_{3}(x)\,(3\text{-th cyclotomic polynomial})$ & $S_3$ \\   
        
        $5$ & $Q_{5}(x)\,(5\text{-th cyclotomic polynomial})$ & $S_5$ \\   

        $7$ & $Q_{7}(x)\,(7\text{-th cyclotomic polynomial})$ & $S_7$ \\   

        $31$ & $Q_{31}(x)\,(31\text{-th cyclotomic polynomial})$ & $S_{31}$ \\   

        $3\cdot 5$ & $(x-1)\cdot Q_{3}(x)\cdot Q_{5}(x)$ & $S_3\times S_{5}$ \\   

        $3\cdot 5$ & $Q_{15}(x)\,(217\text{-th cyclotomic polynomial})$ & $S_3\times S_{5}$ \\   

        $7\cdot 31$ & $(x-1)\cdot Q_{7}(x)\cdot Q_{31}(x)$ & $S_7\times S_{31}$ \\   

        $7\cdot 31$ & $Q_{217}(x)\,(217\text{-th cyclotomic polynomial})$ & $S_7\times S_{31}$ \\   

        $2\cdot 3\cdot 5$ & $Q_{3}(x)\cdot Q_{5}(x)$ & $S_2 \wr (S_3\times S_{5})$ \\   

        $7\cdot 3\cdot 5$ & $Q_{3}(x)\cdot Q_{5}(x)$ & $S_7 \wr (S_3\times S_{5})$ \\   

        $3\cdot 5\cdot 7$ & $Q_{5}(x)\cdot Q_{7}(x)$ & $S_3 \wr (S_5\times S_{7})$ \\
        \bottomrule[1.2pt]
	\end{tabular}
	\label{B_E}
    \end{table*}

Let $g(x)\in\mathbb{F}_{r^\alpha}[x]$ be a factor of $x^p-1$. The contributions of this paper are summarized as follows.
\begin{itemize}
    \item For a positive integer $h$ and a prime number $p$, the permutation group of $\mathcal{C}_{hp,g(x)}$ is completely determined (see Theorem~\ref{L_hp}). 
    \item Let $n$ be a positive number and $m$ be a non-negative number. For any $0\le u\le m$ and $0\le v\le n-1$, the permutation group of $\mathcal{C}_{r^mp^n,g(x^{r^up^v})}$ is completely determined (see Theorem~\ref{L_rp}), and the permutation group of $\mathcal{C}_{hr^mp^n,g(x^{r^up^v})}$ is completely determined (see Corollary~\ref{L_rp_col}). 
    \item For distinct prime numbers $p\neq r$ and $q\neq r$, the permutation groups  of $\mathcal{C}_{pq,Q_{pq}(x)}$ and $\mathcal{C}_{pq,(x-1)Q_p(x)Q_q(x)}$ are completely determined (see Theorem~\ref{L_pq}), and the permutation group of $\mathcal{C}_{hpq,Q_p(x)Q_q(x)}$ are completely determined (see Corollary~\ref{L_pq_col}). Moreover, it is noteworthy that, to our knowledge, this is the first time that the permutation groups of cyclic codes whose generator polynomials are factors of $x^{pq}-1$ but not factors of $x^p-1(\text{or }x^q-1)$ has been obtained.
\end{itemize}
Besides, with the help of~\textsc{Magma}, we present the computational results about certain binary cyclic codes in Table~\ref{B_E}, which are consistent with the results presented in this paper. 

The results of this paper indicated that a strong connection exists between the permutation groups of many cyclic codes with very long lengths and special generator polynomials, and those of cyclic codes of prime lengths. It would be interesting to obtain the permutation groups of cyclic codes with more specific lengths such as $pq^2$, as well as those with more general generator polynomials by using the two matrix representations presented in this paper.

\section*{Acknowledgment}

The first and the third authors are supported by GuangDong Basic and Applied Basic Research Foundation (No. 2025A1515011764), the National Natural Science Foundation of China (No. 12441107) and National Key Research and Development Program of
China (No.2025YFA1017100). The second author is supported by NSF of Chongqing (Grant No. CSTB2025NSCQ-GPX0333). The authors are grateful to Mr. Haojie Chen for his discussions regarding the results in this paper.

\bibliographystyle{ieeetr}
\bibliography{reference}

\end{document}